\documentclass{article}
\usepackage[utf8]{inputenc}
\usepackage{fourier}
\usepackage{amsthm}
\usepackage{amssymb}
\usepackage{amsmath}
\usepackage{mathtools}
\usepackage{bbm}
\usepackage{booktabs}
\usepackage{tikz}
\usepackage{yfonts}
\usepackage{algorithm}
\usepackage{algpseudocode}
\usepackage{float}

\usepackage{comment}

\theoremstyle{plain}
\newtheorem{thm}{Theorem}
\newtheorem{lemma}{Lemma}

\theoremstyle{definition}

\newtheorem{assumption}{Assumption}
\newtheorem{prop}{Proposition}

\newtheorem{example}{Example}

\theoremstyle{remark}
\newtheorem{remark}{Remark}

\newtheoremstyle{note}{3pt}{3pt}{\addtolength{\leftskip}{2.5em}}{}{\bfseries}{:}{0.5em}{}
\theoremstyle{note}

\DeclareMathOperator*{\argmin}{arg\,min}

\title{No Free Lunch for Stochastic Gradient Langevin Dynamics}
\author{Natesh Pillai, Aaron Smith, Azeem Zaman}
\begin{document}
\newcommand{\TV}[2]{\Vert \mathcal{L}(#1) - \mathcal{L}(#2) \Vert_{\text{TV}}}
\maketitle

\begin{abstract}
As sample sizes grow, scalability has become a central concern in the development of Markov chain Monte Carlo (MCMC) methods. One general approach to this problem, exemplified by the popular stochastic gradient Langevin dynamics (SGLD) algorithm, is to use a small random subsample of the data at every time step. This paper, building on recent work such as \cite{nagapetyan2017true,JohndrowJamesE2020NFLf}, shows that this approach often fails: while decreasing the sample size increases the speed of each MCMC step, for typical datasets this is balanced by a matching decrease in accuracy. This result complements recent work such as  \cite{nagapetyan2017true} (which came to the same conclusion, but analyzed only \textit{specific upper bounds on errors} rather than \textit{actual errors}) and \cite{JohndrowJamesE2020NFLf} (which did not analyze nonreversible algorithms and allowed for logarithmic improvements). 
\end{abstract}

\section{Introduction}

It is well-known that MCMC algorithms often perform poorly when applied to large datasets. The simplest cause of this poor scaling is the fact that many popular MCMC algorithms require a computation involving every data point at every time step of the algorithm. This suggests that the per-step computational cost of MCMC scales at least linearly in the size $n$ of the data set. It is natural to try to avoid this problem by looking only at a subsample of the data at every time-step. The goal of this paper (together with the longer companion paper \cite{JohndrowJamesE2020NFLf}) is to give clear limits on the performance of simple subsample-based algorithms. 

The main results of the current paper, Theorems \ref{ThmMainThm} and \ref{ThmMainThm2}, give a strong sense in which improvement is impossible. They show that, under certain conditions, the error of a popular subsampling algorithm cannot be much smaller than the full-sample algorithm it was based on. Our paper gives only \textit{lower} bounds on the error, but effectively-matching upper bounds can be found in \cite{Dalalyan_2019} and the references therein.

The current paper differs from the companion \cite{JohndrowJamesE2020NFLf} in two main ways. The most important is that the companion paper did not apply at all to nonreversible chains, although many of the most popular subsampling chains are not reversible. The current paper closes this gap, showing that the same qualitative conclusions hold for the simplest and most popular nonreversible algorithm when applied to classical statistical models. Secondarily, by focusing on a specific algorithm, the current paper can give assumptions that are much easier to read and verify than the very generic assumptions in \cite{JohndrowJamesE2020NFLf}. 

The rest of the introduction will briefly recall the history of the subject. Early work such as \cite{korattikara2014austerity}, and later work such as \cite{InformedSSSurvey18,MBGibbs18}, showed promising empirical results for subsampling MCMC. This was accompanied by theoretical work such as \cite{RudolfWass15}, \cite{TallData17} and \cite{MBGibbs18} itself. These theoretical papers gave quantitative error bounds on subsampling algorithms, including proofs that the methods could be used to obtain consistent estimates. However, as \cite{TallData17} explained, these early theoretical results seemed to fall short of what was desired: bounds showing that subsampling MCMC was \textit{better} than ``naive" MCMC for a wide range of realistic statistical examples, rather than merely showing that it was \textit{not much worse.}

In \cite{nagapetyan2017true}, the authors presented some evidence that this falling-short reflected a real problem in the underlying algorithm, not merely a limitation in the proof techniques. The paper studied a natural approach to analyzing  a popular subsampling algorithm, called stochastic gradient Langevin dynamics (SGLD). Its main conclusion was that, after appropriate counting of computational costs, these error bounds did not depend on the size of the subsample used. In other words, the error bounds couldn't be used to show that subsampling was an effective strategy. However,  \cite{nagapetyan2017true} left open the possibility that this was an artifact of their analysis of the algorithm - perhaps subsampling was effective, even if their error bounds couldn't show it.

In \cite{JohndrowJamesE2020NFLf}, some of the present authors showed that this was not the case for a wide variety of algorithms: subsampling MCMC really could not offer substantial speedups. However, \cite{JohndrowJamesE2020NFLf} still had a number of important gaps. The most important was related to reversibility: the strongest conclusions in \cite{JohndrowJamesE2020NFLf}  only apply to reversible chains. Since the SGLD algorithm studied in \cite{nagapetyan2017true} is nonreversible, this means that there was still an open question as to whether SGLD could result in substantial improvement in some regime.

The main contributions of this paper are short proofs that this doesn't happen: subsampling can't improve certain important aspects of the performance of SGLD. Informally, Theorem \ref{ThmMainThm} says the following: there are many posterior distributions for which \textit{any} SGLD algorithm will return extremely poor samples \textit{until} it has seen each data point at least once on average. Theorem \ref{ThmMainThm2} gives analogous results in the regime that all data points have been seen many times on average, showing that SGLD is not much more efficient than the analogous algorithm that looks at all datapoints at all time steps.

This paper is a short note aimed at closing an important gap left by \cite{JohndrowJamesE2020NFLf}  and \cite{nagapetyan2017true}, and so we make no effort to give results that are as general as possible. For readers who might be interested in obtaining similar results for other non-reversible chains, we note that the key technical idea was to use the coupling construction given in Section \ref{SecCoupConSGLDExp} on a ``reasonable" forward-mapping representation of SGLD as given in Example \ref{ExForwardSGLD}. Exactly the same coupling construction can be used for \textit{any} non-reversible chain with the same type of driving randomness, and this can be used to obtain small perturbations for the underlying Markov chain. Perturbation bounds for the Markov chain can then be combined with perturbation bounds for specific models, as described in the companion paper \cite{JohndrowJamesE2020NFLf}.

\section{Notation and Generic Bounds} \label{SecNotGenBound}

We introduce generic notation and bounds for MCMC algorithms. These will be applied to specific target distributions later in the paper.

\subsection{Basic Notation}

\begin{comment}
For $p \geq 1$, the $p$-Wasserstein distance between probability measures $\mu$ and $\nu$ on a metric space $(\Omega,d)$ is given by
\begin{align}
    W_p(\mu,\nu) = \left(\inf_{\gamma \in \Gamma(\mu,\nu)} \int d(x,y)^p \, d\gamma(x,y)\right)^{1/p} = \left(\inf_{\gamma \in \Gamma(\mu,\nu)} \mathbb{E}_{\gamma}\left[d(X,Y)^p\right]\right)^{1/p}. \label{eq:defn-Wasserstein}
\end{align}
where $\Gamma(\mu,\nu)$ is the set of all couplings of $\mu$ and $\nu$ and $(X,Y) \sim \gamma$ in the expectation. 

Recall that a function $f \, : \, (\Omega,d) \mapsto \mathbb{R}^{d}$ is call $K$-Lipschitz if
\begin{equation*}
    \frac{\|f(x)-f(y)\|}{d(x,y)} \leq K
\end{equation*}
for all $x,y \in \Omega$. Denote by $Lip(\Omega)$ the collection of 1-Lipschitz functions.
\end{comment}

For two probability distributions $\mu, \nu$ on a common Polish probability space $(\Omega, \mathcal{F})$, the total variation distance is defined by:
\begin{equation*}
    \| \mu - \nu \|_{\mathrm{TV}} = \sup_{A \in \mathcal{F}}|\mu(A) - \mu(B)|.
\end{equation*}

For $n \in \mathbb{N}$, define $[n] = \{1,2,\ldots,n\}$. By a small abuse of notation, we will sometimes treat a vector as a set (writing e.g. $x \in v$ rather than $x \in \{v_{1},\ldots,v_{n}\}$) when there is no possible ambiguity. We will frequently consider datasets $\{ X_{1},\ldots,X_{n} \}$, and use $\mathcal{X}_{n} = \{ X_{1},\ldots,X_{n} \}$ to denote the full dataset.

\subsection{Forward-Mapping Representation of Subsampling MCMC}

We consider the usual setup for Bayesian inference. We have a parameter space $(\Theta, \mathcal{F}_{\Theta})$, a family of models $\{p(\cdot|\theta)\}_{\theta \in \Theta}$, a prior $p$, and a dataset $\mathcal{X} = \{X_1,\ldots, X_n\}$. For the entire paper, we assume that $(\Theta,d)$ is a Polish space and that $\mathcal{F}_{\Theta}$ is the usual Borel $\sigma$-algebra.

We now define a family of ``subsampling" MCMC algorithms, which we will show contains many popular algorithms as special cases.  

Fix a pair of probability spaces $(\mathbb{A},\mathcal{F}_{\mathbb{A}})$ and $(\Omega, \mathcal{F}_{\Omega})$ and a maximum ``subsample size" $m \in \mathbb{N}$. Let $F \, : \, \Theta \times \mathbb{A} \times [n]^{m} \times \Omega \mapsto \Theta$ be a measurable function, let $\mu$ be a probability measure on $[n]^{m}$, and let $\eta$ be a probability measure on $\Omega$.\footnote{This last is used to capture all of the random variables sampled during a step of the underlying algorithm, \textit{except} the choice of subsample. Typically, we can assume WLOG that this is of the form $[0,1]^{k}$ with the usual Borel $\sigma$-algebra.} These choices together define a forward-mapping representation for a Markov chain $\{\theta_{t}\}_{t \geq 0}$ via Algorithm \ref{alg:gen-alg-form}: 

\begin{algorithm}[H]
\caption{Sampling with explicit randomness} \label{alg:gen-alg-form}
\begin{algorithmic}[1]
\State \textbf{Input:} Starting point $(\theta_0, r_{0}) \in \Theta \times \mathbb{A}$, number of iterations $T \in \mathbb{N}$.
\For {$t = 1,\ldots T$}
    \State Sample $E_{t} \sim \mu$, $U_{t} \sim \eta$.
    \State Update state $(\theta_t, r_t) \gets F( \theta_{t-1}, r_{t-1}, E_t, U_{t})$
\EndFor
\State Return the samples $\{\theta_1,\ldots, \theta_T\}$ 
\end{algorithmic}
\end{algorithm}

One can check that this generic algorithm covers many popular subsampling algorithms, including stochastic gradient Langevin dynamics (see \cite{WellingTehSGLD11}) and subsampling pseudo-marginal Metropolis-Hastings (see \cite{quiroz2015speeding}). The auxiliary variable $r_{t}$ allows us to include popular algorithms such as Stochastic Gradient Hamiltonian Monte Carlo \cite{ChenTianqi2014SGHM} or the pseudomarginal algorithm \cite{andrieu2009pseudo} in the form of Algorithm \ref{alg:gen-alg-form}. 

As a concrete example, the following gives a representation of the usual SGLD algorithm on $\mathbb{R}^{d}$:

\begin{example} \label{ExForwardSGLD}

Fix dimension $d \in \mathbb{N}$, step-size $\epsilon > 0$ and minibatch-size $m = m(\theta) \leq M$. Then SGLD can be written in the form of Algorithm \ref{alg:gen-alg-form} with the choices:
\begin{equation} \label{Eq_SGLD_Def1}
\mathbb{A} = \emptyset, \qquad \Omega = \mathbb{R}^{d}, \qquad \eta = N(0,\mathbf{1}_{d}), \qquad \mu = \mathrm{Unif}[[n]^{M}]
\end{equation}

and

\begin{equation} \label{Eq_SGLD_Def2}
    F(\theta,r,E,U) = \theta + \frac{\epsilon}{2}\left(\nabla \log p(\theta) + \frac{n}{m(\theta)}\sum_{j = 1}^{m(\theta)} \nabla\log p(X_{E[j]} \mid \theta)\right) + U.
\end{equation}

\end{example}

The main constraint comes from the restriction to i.i.d. selection of minibatches, which excludes some algorithms with complex minibatch selection mechanisms (see \textit{e.g.} \cite{InformedSSSurvey18} for a survey of approaches). We will show in later sections how some algorithms that do this can still be placed in the framework of this paper (see Remark \ref{RemComplexSelectionMsr}). 
\subsection{Perturbations of Posterior Distributions of Exponential Families} \label{SecPertPost}

We focus exponential families with likelihoods of the form 
\begin{equation} \label{DefExpFamily} 
p(x \mid \theta) = h(x)\exp(\theta R(x) - A(\theta)),
\end{equation}
where $x \in \mathbb{R}$, $h$ and $R$ are real-valued functions, and $A$ is the associated log-normalizing constant. Fixing a prior $p$ on $\Theta$, the associated posterior density is given by
\begin{align}
    p(\theta \mid \mathcal{X}_{n}) = \frac{\exp(\theta S - nA(\theta))p(\theta)}{\int \exp(\theta' S - nA(\theta')) p(\theta')\,d\theta'}, \label{eq:target-posterior}
\end{align}
where $S = \sum_{i=1}^n R(X_i)$ is the sufficient statistic.

We then consider the ``perturbed" posterior associated with a ``perturbed" sufficient statistic $S + \delta$:
\begin{align}
    p_\delta(\theta \mid \mathcal{X}_{n}) = \frac{\exp(\theta (S + \delta) - nA(\theta))p(\theta)}{\int \exp(\theta' (S + \delta) - nA(\theta')) p(\theta')\,d\theta'}. \label{eq:perturbed-posterior}
\end{align}

We will use the following perturbation assumption on the likelihood when analyzing MCMC on ``small" time intervals: 

\begin{assumption} \label{AssumptionCGE_Good}

There exists $\gamma > 0$ so that for all sequences $c_{n} \rightarrow \infty$, all datasets $X_{1},X_{2},\ldots$, and all sequences $\delta_{n} > \frac{c_{n}}{\sqrt{n}}$,
\begin{align*}
\liminf_{n \rightarrow \infty} \| p(\theta \mid \mathcal{X}_{n}) -  p_{\delta_{n}}(\theta \mid \mathcal{X}_{n}) \|_{\mathrm{TV}} \geq  \gamma.
\end{align*} 
\end{assumption}

This condition says that the posterior distribution is somewhat sensitive to large changes in the data. This assumption holds with $\gamma = 1$ when the posterior distribution is close to Gaussian, as well as in many other situations. See Appendix \ref{sec:technical-lemmas} for a simple sufficient condition for Assumption \ref{AssumptionCGE_Good} that is based only on the first two moments of the posterior. Appendix \ref{sec:technical-lemmas} also contains an application of this sufficient condition to show that the Gaussian model satisfies this condition with $\gamma = 1$. 

Assumption \ref{AssumptionCGE_Good} is relevant to small time scales. On large time scales, we consider smaller perturbations, leading to the assumption:

\begin{assumption} \label{AssumptionCGE_Good_Big_Time}

There exists $\gamma > 0$ so that for all sequences $c_{n} \rightarrow 0$, all datasets $X_{1},X_{2},\ldots$, and all sequences $\delta_{n} < \frac{c_{n}}{\sqrt{n}}$,
\begin{align*}
\liminf_{n \rightarrow \infty} \| p(\theta \mid \mathcal{X}_{n}) -  p_{\delta_{n}}(\theta \mid \mathcal{X}_{n}) \|_{\mathrm{TV}} \geq  \gamma \delta_{n}.
\end{align*} 
\end{assumption}

Again, this assumption is easy to verify for posterior distributions that are near-Gaussian, and it holds in substantial generality. See Appendix \ref{secAppLongTime} for a quick proof and related calculations.

\section{Perturbation Lower Bounds for Stochastic Gradient Langevin Dynamics Targeting Exponential Families}

We introduce our coupling construction, then prove our main result, Theorem \ref{ThmMainThm}.

\subsection{Coupling Construction for SGLD} \label{SecCoupConSGLDExp}

We consider the representation of SGLD given in Example \ref{ExForwardSGLD}. In this section, we construct a coupling of two Markov chains to ensure that they remain far apart with high probability. This construction represents the main observation behind our main result is: it is possible to couple the choice of subsamples across two chains so that the two chains are \textit{quite close in total variation,} but certain functionals are \textit{somewhat far in expectation.}

%For the purpose of establishing a perturbation for exponential families, we focus on the case where $\mathcal{S}_1(E_t,t) = m$ for some positive integer $m$. We further assume that the selection function $S_2$ is designed to draw $m$ independent samples uniformly from $\mathcal{X}$

We begin by constructing a distribution $\nu$  that is close (in total variation distance) to the distribution $\mu$ used in Example \ref{ExForwardSGLD}. Fix a parameter $\delta > 0$ and $M,n \in \mathbb{N}$. Sample $E[1],\ldots,E[M] \stackrel{i.i.d.}{\sim} \mathrm{Unif}(\{1,2,\ldots,n\})$, then independently sample $B[1],\ldots,B[m]  \stackrel{i.i.d.}{\sim} \mathrm{Bern}(\delta)$ and $E^{+}[1],\ldots,E^{+}[M]  \stackrel{i.i.d.}{\sim} \mathrm{Unif}(\{\lceil \frac{n}{2} \rceil,\ldots,n\}) $. Finally, for $i \in \{1,2,\ldots,M\}$, set:

\[ 
D[i] = \begin{cases} 
       E[i], &\qquad B[i] = 0 \\
      E^{+}[i], &\qquad B[i] = 1. \\
   \end{cases}
\]

Denote by $\nu = \nu_{\delta,M}$ the distribution of the vector $D$ and $\mu = \mu_{M}$ the distribution of the vector $E$.

Recall that sampling $E_{t} \sim \mu_{M} \equiv \mathrm{Unif}[[n]^{M}]$ gives the SGLD algorithm in Equations \eqref{Eq_SGLD_Def1} to \eqref{Eq_SGLD_Def2}, targeting the posterior $p(\theta | x)$. If we were to replace $\mu_{M}$ by $\nu_{M,\delta}$ in the algorithm, this construction turns out to give another valid SGLD algorithm - but targeting a certain ``weighted" posterior distribution. More precisely, 

\begin{lemma}[Weighted gradient estimation] \label{lem:weighted-gradient-estimation}
Fix $0 < \delta <1$, sample the set $D \sim \nu_{M,\delta}$, and define $w_i = \mathbb{P}(D[1] = i)$ for $i \in \{1,\ldots, n\}$. Then
\begin{align*}
   \mathbb{E}\left[ \frac{n}{m}\sum_{j \in D} \nabla \log p(X_{j} \mid \theta) \right] = \sum_{i=1}^{n} n \, w_i\nabla \log p(X_i \mid \theta).
\end{align*}
\end{lemma}

\begin{proof}
For all $j$, we have by the definition of $w_i$ that
\begin{align*}
    \mathbb{E}[n \nabla \log p(X_{D[1]} \mid \theta)] = \sum_{i=1}^n n \,w_i\nabla \log p(X_i \mid \theta).
\end{align*}

The result follows from linearity of expectation. 
\begin{comment}
    By the assumption that $D_{tj}$ is independent of $S_1(D_t,t)$, for any $m$ we have
\begin{align*}
    &\mathbb{E}\left[ \frac{n}{S_1(D_t,t)}\sum_{j = 1}^{S_1(D_t,t)} \nabla \log p(X_{\phi(D_{tj})} \mid \theta) \mid S_1(D_t,t) = m\right] = \\
    &\quad \frac{1}{m}\sum_{j=1}^m \mathbb{E}[n \nabla \log p(X_{\phi(D_{tj})} \mid \theta) \mid S_1(D_t, t) = m] = \\
    &\quad \sum_{i=1}^n nw_i\nabla \log p(X_i \mid \theta).
\end{align*}
As the estimator is unbiased conditional on all possible values of $m$, the unconditional expectation is also unbiased. 

\end{comment}

\end{proof}

We interpret this as saying that $ \frac{n}{m}\sum_{j \in D} \nabla \log p(X_{j} \mid \theta)$ is an unbiased estimator for the gradient of the log-likelihood for the weighted model $\prod_{i=1}^n p(X_i \mid \theta)^{nw_i}$.

\begin{comment}

When the driving randomness follows a standard exponential distribution, we have that $w_i = 1/n$ and this result reduces to the unweighted likelihood. For any driving randomness $D_t$, this result allows us to characterize the posterior distribution that SGLD driven by $D_t$ approximates. Chain 2 driven by $D_1,\ldots, D_T$ will target a posterior distribution with the same prior and the weighted likelihood given in Lemma \ref{lem:weighted-gradient-estimation}. The distance between this weighted posterior and the original posterior depends on the how much the weights $w_i$ are perturbed from the uniform distribution where all weights are equal to $1/n$. 

\end{comment}

We can now couple samples from the original and ``perturbed" SGLD algorithm:

\begin{lemma} \label{LemmaChatPertDrive}
There exists a universal constant $C$ with the following property. Fix $M \in \mathbb{N}$, $0 < \alpha < \frac{1}{2}$, $0 < \delta < \frac{1}{2}\min \left(1,\frac{C \alpha^{2}}{ \sqrt{M}}\right)$. Then 
\begin{align*}
\| \mu_{M} - \nu_{M,\delta} \|_{TV}  \leq \alpha.
\end{align*}
\end{lemma}

\begin{proof} 

We use a calculation that is very similar to the proof in Section 1.3 of \cite{chatterjee2019general}.  Let $D \sim \mu_{M}$ and $E \sim \nu_{M,\delta}$. Let $s =\lceil \frac{n}{2} \rceil$, let $p = P[E[1] \geq s] - \frac{1}{2}$ and $q = \frac{1}{2} - P[E[1] < s]$ be the ``bias" induced by $\nu_{M,\delta}$. Let $D^{+} = \{i \in [M] \, : \, D[i] \geq s\}$ and $E^{+} = \{i \in [M] \, : \, E[i] \geq s\}$. 

We are interested in counting the number of times that points with large index are chosen. Let $f$ be the density of the distribution of the distribution $\mathcal{L}(|D^{+}|)$  with respect to $\mathcal{L}(|E^{+}|)$. We then have:
\begin{align*} 
f(D) = (\frac{1}{2} + p)^{D^{+}} (\frac{1}{2}-q)^{M-D^{+}} = [(\frac{1}{2}+p)(\frac{1}{2}-q)]^{\frac{M}{2}} \left(\frac{1+p}{1-q}\right)^{D^{+}-\frac{M}{2}}.
\end{align*}

Thus,
\begin{align*}
\| \mu_{M} - \nu_{M,\delta} \|_{TV} \leq \mathbb{E}\left[\left(1 - [(\frac{1}{2}+p)(\frac{1}{2}-q)]^{\frac{M}{2}} \left(\frac{1+np}{1-nq}\right)^{D^{+}-\frac{M}{2}} \right)_{+} \right].    
\end{align*}

Note that $D^{+}$ is a binomial random variable. Applying Chebyshev's inequality to $D^{+}$ with the scaling $\delta = \frac{\alpha}{n \sqrt{M}}$, the same calculation as in Section 1.3 of \cite{chatterjee2019general} says that there exists a universal constant $C > 0$ so that

\begin{align*}
\| \mu_{M} - \nu_{M,\delta} \|_{TV} \leq C \sqrt{\alpha}.
\end{align*}

This completes the proof.
\end{proof}

\begin{remark} [Complex Selection Measures] \label{RemComplexSelectionMsr}
Note that the proof of Lemma \ref{LemmaChatPertDrive} depended entirely on bounding the density $f$. In particular, the same argument can be used with minimal changes even when the driving measure $\mu_{M}$ is not uniform. We don't attempt to give universal estimates here, as the literature on SGLD contains a large variety of driving measures with rather different behaviors.
\end{remark}

We now wish to establish that this small perturbation of the driving randomness corresponds to a similarly-sized perturbation of the sufficient statistics in the effective model. We have the following:

\begin{lemma} \label{lem:expectation-perturbation-lower-bound}
Fix $\delta > 0$, $M \in \mathbb{N}$ and a non-decreasing sequence $X_{1} \leq X_{2} \leq \ldots \leq X_{n}$. Let $s = \lceil \frac{n}{2}\rceil$. Let $D \sim \mu_{M}$ and $E \sim \nu_{M, \delta}$. Then 

\begin{align*}
\mathbb{E}[\sum_{i \in E} R(X_{i}) - \sum_{i \in D} R(X_{i}) ] = \delta \left(\frac{1}{n-s+1} \sum_{i = s}^{n} R(X_{i}) - \frac{1}{n} \sum_{i =1}^{n} R(X_{i}) \right).
\end{align*}
\end{lemma}

\begin{proof}
    This follows immediately from the definition of $\mu_{M}, \nu_{M,\delta}$.
\end{proof}

We add the (very weak) assumption that this is large with high probability:

\begin{assumption} \label{AssumptionOrderMatters}
We say that a likelihood of the form \eqref{DefExpFamily} and parameter $\theta_{0}$ are $\eta$-good if they have the following property. 

Fix $n \in \mathbb{N}$, let $s = \lceil \frac{n}{2} \rceil$, let $Y_{1},Y_{2},\ldots, Y_{n} \stackrel{i.i.d.}{\sim} p(\cdot \mid \theta_{0})$, and let $X_{1} \leq X_{2} \leq \ldots \leq X_{n}$ be the same points put in ascending order. Then 
\begin{equation*}
\lim_{n \rightarrow \infty} \mathbb{P}\left[\left| \frac{1}{n-s+1} \sum_{i = s}^{n} R(X_{i}) - \frac{1}{n} \sum_{i =1}^{n} R(X_{i}) \right| < \eta n\right] = 0.
\end{equation*}
\end{assumption}

This assumption is quite weak, and is straightforward to check for popular functions $R$. To give a concrete example, when $R(x) = x$, this follows immediately from Chebyshev's inequality.

\subsection{Proof of Main Result: Short Runs}

We set notation for the main result. Fix a prior $p$ and model $p(\cdot \mid x)$. For fixed dataset $X$ and sample size $M$, Denote by $K_{X,M}$ the transition kernel associated with the SGLD algorithm in Example \ref{ExForwardSGLD}. Similarly, for fixed $\delta > 0$, denote by $K_{X,\delta,M}$ the transition kernel associated with the SGLD algorithm in Example \ref{ExForwardSGLD}, with $\mu_{M}$ replaced by $\nu_{\delta,M}$.

Finally, we set sequences of integers $T_{n},M_{n},\omega_{n}$ satisfying
\begin{equation} \label{EqTMCond} 
\lim_{n \rightarrow \infty} \frac{T_{n} M_{n}}{n} = 0
\end{equation} 
and
\begin{equation} \label{EqTMOCond}
\lim_{n\rightarrow \infty} \omega_{n} \sqrt{M_{n} T_{n}} = 0, \qquad \lim_{n\rightarrow \infty} \omega_{n} \sqrt{n} = \infty.
\end{equation}

%Finally, denote by $\pi_{X,M}$ and $\pi_{X,\delta,M}$ their respective stationary measures.

Our main result is:

\begin{thm}\label{ThmMainThm}
Fix a prior $p$, likelihood $p(\cdot \mid \theta)$ and parameter value $\theta_{0}$ satisfying Assumptions \ref{AssumptionCGE_Good} and  \ref{AssumptionOrderMatters}.  Let $X_{1},X_{2},\ldots \stackrel{i.i.d.}{\sim} p(\cdot \mid \theta_{0})$ be data sampled from this. Let $T_{n}, M_{n}, \omega_{n}$ satisfy Equations \eqref{EqTMCond}  and \eqref{EqTMOCond}. Let $Z_{1}^{(n)},Z_{2}^{(n)},\ldots,Z_{T_{n}}^{(n)}$ be a Markov chain sampled from $K_{M_{n},\mathcal{X}_{n}}$. Similarly, let $\tilde{Z}_{1}^{(n)},\ldots,\tilde{Z}_{T_{n}}^{(n)}$ be sampled from $K_{\mathcal{X}_{n}, \omega_{n}, M_{n}}$. Then for all $a > 0$:
\begin{equation*}
\lim_{n \rightarrow \infty} \mathbb{P}[\|\mathcal{L}(Z_{T_{n}}^{(n)}) - p(\cdot | \mathcal{X}_{n})\|_{\mathrm{TV}} + \|\mathcal{L}(\tilde{Z}_{T_{n}}^{(n)}) - p_{\omega_{n}}(\cdot | \mathcal{X}_{n})\|_{\mathrm{TV}} )\geq (1-a)\gamma] = 1.
\end{equation*}
\end{thm}

\begin{remark}
Our notation suppresses certain dependencies, such as the step-size and the distribution of the first point. This should be read as shorthand for the following: the suppressed parameters of the algorithm should be taken to be a function of \textit{only} the size $n$ of the data, but they may be \textit{any} such function. In particular, this means that e.g. the starting distribution for the chain $Z_{0}^{(n)}$ in the theorem statement could be the target $p(\cdot | \mathcal{X}_{n})$ for each $n \in \mathbb{N}$ - but then the starting distribution for the chain $\tilde{Z}_{0}^{(n)}$ would also be the target $p(\cdot | \mathcal{X}_{n})$ for each $n \in \mathbb{N}$. This is a natural assumption to make, as tuning parameters are often strongly dependent on the size $n$ of the dataset, and we don't wish to restrict the allowed tuning parameters.
\end{remark}
\begin{proof}
For fixed $n$, denote by $\mathcal{E}_{n}$ the event that 
\begin{equation*} 
\left| \frac{1}{n-s+1} \sum_{i = s}^{n} R(X_{i}) - \frac{1}{n} \sum_{i =1}^{n} R(X_{i}) \right| \geq \eta n,
\end{equation*}
where $s = \lfloor \frac{n}{2} \rfloor$. 

We first check that the target distributions of our algorithms are far. 
Combining Lemma \ref{lem:weighted-gradient-estimation} and Lemma \ref{lem:expectation-perturbation-lower-bound}, on the event $\mathcal{E}_{n}$ we have 
\begin{equation} \label{IneqTargetsFar}
\liminf_{n \rightarrow \infty} \| p(\cdot | \mathcal{X}_{n}) - p_{\omega_{n}}(\cdot | \mathcal{X}_{n}) \|_{\mathrm{TV}} \geq \gamma.
\end{equation}

We next check that the samples of our algorithms are close. By Lemma \ref{LemmaChatPertDrive}, it is possible to sample a Markov chain $\tilde{Z}_{1}^{(n)},\ldots,\tilde{Z}_{T_{n}}^{(n)}$ from $K_{\mathcal{X}_{n}, \omega_{n}, M_{n}}$ so that 
\begin{equation} \label{IneqCoupGood}
\lim_{n \rightarrow \infty} \mathbb{P}[(Z_{1}^{(n)},\ldots,Z_{T_{n}}^{(n)}) = (\tilde{Z}_{1}^{(n)},\ldots,\tilde{Z}_{T_{n}}^{(n)})]  = 1.
\end{equation}

Next, by the triangle inequality,

\begin{align*}
\Delta_{n} &\equiv \|\mathcal{L}(Z_{T_{n}}^{(n)}) - p(\cdot | \mathcal{X}_{n})\|_{\mathrm{TV}} +  \|\mathcal{L}(\tilde{Z}_{T_{n}}^{(n)}) - p_{\omega_{n}}(\cdot | \mathcal{X}_{n})\|_{\mathrm{TV}} \\
& \geq \| p(\cdot | \mathcal{X}_{n}) - p_{\omega_{n}}(\cdot | \mathcal{X}_{n})\|_{\mathrm{TV}} - \| \mathcal{L}(Z_{T_{n}}^{(n)}) - \mathcal{L}(\tilde{Z}_{T_{n}}^{(n)}) \|_{\mathrm{TV}}.
\end{align*}

Applying Inequalities \eqref{IneqCoupGood} and \eqref{IneqTargetsFar} to the terms on the right-hand side of this inequality, we have for all $a > 0$

\begin{align*}
\lim_{n \rightarrow \infty} \mathbb{P}[\Delta_{n} <  (1-a)\gamma] \leq \lim_{n \rightarrow \infty} \mathbb{P}[\mathcal{E}_{n}^{c}].
\end{align*}

By Assumption \ref{AssumptionOrderMatters}, this last limit is 0, completing the proof.
\end{proof}

\begin{remark} [Why Does Theorem \ref{ThmMainThm} Bound a Sum of Errors?]
We note that Theorem \ref{ThmMainThm} gives a lower bound on the sum of errors for two Markov chains, while it is more common to give a lower bound on the error of a single Markov chain. This is not an accident, and indeed it is completely unavoidable given our assumptions: we allow the starting measure $\mathcal{L}(X_{0})$ to be the target distribution, and we allow the number of steps $T$ to be 0, so in fact one of the errors may be exactly 0! 

Of course, this is a rather degenerate situation. More broadly, we point out that it is quite possible for a particular procedure to give the right answer for a particular problem essentially ``by accident." Thus, we show that while it is possible for \textit{one} chain to get the right answer, it is not possible for \textit{both} chains to get the right answer. In practice, a user will never know which situation they are in, and so will not be able to take advantage of any ``accidental" accuracy of this type.
\end{remark}

\subsection{Proof of Main Result: Long Runs}

We use the same notation as immediately precedes the statement of Theorem \ref{ThmMainThm}, but replace conditions \eqref{EqTMCond} and \eqref{EqTMOCond} by: 

\begin{equation} \label{EqTMCond2} 
\lim_{n \rightarrow \infty} \frac{T_{n} M_{n}}{n} = \infty
\end{equation} 
and
\begin{equation} \label{EqTMOCond2}
\lim_{n\rightarrow \infty} \omega_{n} \sqrt{M_{n} T_{n}} = 0.
\end{equation}

The proof is quite similar to the proof of Theorem \ref{ThmMainThm}, though slightly more attention must be paid to the relative rates at which different errors go to 0.

\begin{thm}\label{ThmMainThm2}
Fix a prior $p$, likelihood $p(\cdot \mid \theta)$ and parameter value $\theta_{0}$ satisfying Assumptions \ref{AssumptionCGE_Good} and  \ref{AssumptionOrderMatters}.  Let $X_{1},X_{2},\ldots \stackrel{i.i.d.}{\sim} p(\cdot \mid \theta_{0})$ be data sampled from this. Let $T_{n}, M_{n}, \omega_{n}$ satisfy Equations \eqref{EqTMCond}  and \eqref{EqTMOCond}. Let $Z_{1}^{(n)},Z_{2}^{(n)},\ldots,Z_{T_{n}}^{(n)}$ be a Markov chain sampled from $K_{M_{n},\mathcal{X}_{n}}$. Similarly, let $\tilde{Z}_{1}^{(n)},\ldots,\tilde{Z}_{T_{n}}^{(n)}$ be sampled from $K_{\mathcal{X}_{n}, \omega_{n}, M_{n}}$. Then for all $a > 0$:
\begin{equation*}
\lim_{n \rightarrow \infty} \mathbb{P}[\|\mathcal{L}(Z_{T_{n}}^{(n)}) - p(\cdot | \mathcal{X}_{n})\|_{\mathrm{TV}} + \|\mathcal{L}(\tilde{Z}_{T_{n}}^{(n)}) - p_{\omega_{n}}(\cdot | \mathcal{X}_{n})\|_{\mathrm{TV}} )\geq (1-a)\eta \omega_{n}] = 1.
\end{equation*}
\end{thm}

\begin{proof}

For fixed $n$, denote by $\mathcal{E}_{n}$ the event that 
\begin{equation*} 
\left| \frac{1}{n-s+1} \sum_{i = s}^{n} R(X_{i}) - \frac{1}{n} \sum_{i =1}^{n} R(X_{i}) \right| \geq \eta n,
\end{equation*}
where $s = \lfloor \frac{n}{2} \rfloor$.

We first check that the target distributions of our algorithms are far. 
Combining Lemma \ref{lem:weighted-gradient-estimation} and Lemma \ref{lem:expectation-perturbation-lower-bound}, on the event $\mathcal{E}_{n}$ we have 
\begin{equation} \label{IneqTargetsFar2}
\liminf_{n \rightarrow \infty} \, \omega_{n}^{-1} \| p(\cdot | \mathcal{X}_{n}) - p_{\omega_{n}}(\cdot | \mathcal{X}_{n}) \|_{\mathrm{TV}}  \geq \eta \, \gamma.
\end{equation}

We next check that the samples of our algorithms are close. By Lemma \ref{LemmaChatPertDrive}, it is possible to sample a Markov chain $\tilde{Z}_{1}^{(n)},\ldots,\tilde{Z}_{T_{n}}^{(n)}$ from $K_{\mathcal{X}_{n}, \omega_{n}, M_{n}}$ so that
\begin{equation} \label{IneqCoupGood2}
\limsup_{n \rightarrow \infty} \, \omega_{n}^{-\frac{1}{2}} (M_{n} T_{n})^{-\frac{1}{4}} \mathbb{P}[(Z_{1}^{(n)},\ldots,Z_{T_{n}}^{(n)}) = (\tilde{Z}_{1}^{(n)},\ldots,\tilde{Z}_{T_{n}}^{(n)})]  \leq C^{-1},
\end{equation}
where $C$ is as in the statement of Lemma \ref{LemmaChatPertDrive}. Next, by the triangle inequality,

\begin{align*}
\Delta_{n} &\equiv \|\mathcal{L}(Z_{T_{n}}^{(n)}) - p(\cdot | \mathcal{X}_{n})\|_{\mathrm{TV}} +  \|\mathcal{L}(\tilde{Z}_{T_{n}}^{(n)}) - p_{\omega_{n}}(\cdot | \mathcal{X}_{n})\|_{\mathrm{TV}} \\
& \geq \| p(\cdot | \mathcal{X}_{n}) - p_{\omega_{n}}(\cdot | \mathcal{X}_{n})\|_{\mathrm{TV}} - \| \mathcal{L}(Z_{T_{n}}^{(n)}) - \mathcal{L}(\tilde{Z}_{T_{n}}^{(n)}) \|_{\mathrm{TV}}.
\end{align*}

Applying Inequalities \eqref{IneqCoupGood2} and \eqref{IneqTargetsFar2} to the terms on the right-hand side of this inequality, we have for all $a > 0$

\begin{align*}
\lim_{n \rightarrow \infty} \mathbb{P}[\Delta_{n} < \eta \omega_{n} - C^{-1} \omega_{n}^{\frac{1}{2}}(M_{n} T_{n})^{\frac{1}{4}} ] \leq \lim_{n \rightarrow \infty} \mathbb{P}[\mathcal{E}_{n}^{c}].
\end{align*}

By Assumption \ref{AssumptionOrderMatters}, this last limit is 0, so for all $ a > 0$

\begin{align*}
\lim_{n \rightarrow \infty} \mathbb{P}[\Delta_{n} < (1-a)\eta \omega_{n} + (a \eta \omega_{n} - C^{-1} \omega_{n}^{\frac{1}{2}} (M_{n} T_{n})^{\frac{1}{4}}) ] = 0.
\end{align*}

By assumption \eqref{EqTMOCond2}, 

\begin{equation*}
\liminf_{ n \rightarrow \infty} (a \eta \omega_{n} - C^{-1} \omega_{n}^{\frac{1}{2}} (M_{n} T_{n})^{\frac{1}{4}}) \geq 0,
\end{equation*}
so this completes the proof.
\end{proof}

\begin{comment}
\section{Bounds for optimization}

\textcolor{red}{[AMS: I don't know if this got anywhere - ignoring for now.]}

Suppose that we have two optimization problems:
\begin{align*}
    \theta^* = \argmin_{\theta} \sum_{i=1}^n \frac{1}{n}\nabla f(\theta \mid x_i) \quad \quad 
    \theta_{w}^* = \argmin_{\theta} \sum_{i=1}^n w_i\nabla f(\theta \mid x_i).
\end{align*}

\begin{assumption}[Smoothness]
The function $f(\theta mid x)$ is $M$-smooth for $M > 0$:
\begin{align*}
    \Vert \nabla f(\theta \mid x) - \nabla f(\theta \mid y) \Vert_2 \leq M\Vert x - y\Vert_2 \quad \text{for all $x$ and $y$}.
\end{align*}
\end{assumption}

\end{comment}

\section{Discussion}

We have shown that, under certain common conditions, SGLD is not much more efficient than the ``full" gradient Langevin dynamics (usually called ULA in the MCMC literature). We comment on a few questions about subsampling MCMC that are left open by this work:

\begin{enumerate}
    \item \textbf{Filling in the blanks:} Assumption \ref{AssumptionOrderMatters} is quite a strong assumption about the form of the underlying model, and clearly we expect a similar phenomenon to hold in much greater generality. We make two comments:
    \begin{enumerate}
        \item \textbf{One Sufficient Step:} Inspecting the proof of Theorem \ref{ThmMainThm}, we see that Assumption is only used when checking Inequality \eqref{IneqTargetsFar}.  Inequality \eqref{IneqTargetsFar} itself is a basic fact about the target distribution: that ``typical" small tweaks to the distribution of the observed data result in ``similarly-sized" changes to the posterior. Checking this fact in great generality is somewhat difficult and so beyond the scope of a short paper. However, we expect this fact to be true fairly generically. Most of our companion paper  \cite{JohndrowJamesE2020NFLf} is spent on proving analogous (if more complicated) facts in much more general settings. In addition, it is straightforward to estimate the size of the perturbation in  Inequality \eqref{IneqTargetsFar}  from data in low-dimensional problems.
        \item \textbf{Practical Conclusions:} There is a natural concern that a  theoretical result proved only for simple distributions (such as Gaussians) may not extend to more reasonable problems. This skepticism is very reasonable for ``positive" results, such as showing that an estimator converges \textit{quickly} - a method might be taking advantage of the special properties of the simple distribution. However, the main result in this paper is ``negative" - it shows that an estimator is \textit{slow} for some class of simple distributions. We think it would be quite surprising if generic methods such as SGLD happened to fail only on a small class of simple examples, and we are not aware of any empirical work suggesting that this phenomenon happens.
    \end{enumerate}
    
    \item \textbf{Additional Pre-Calculations:} In practice, subsampling is often combined with pre-computations (from simple summary statistics as in \cite{TallData17} to complex objects in \cite{NEURIPS2022_83b17fb3}). It is natural to ask if these pre-calculations can ``save" subsampling algorithms. The answer is somewhat subtle. Certain sophisticated pre-calculations \textit{can} lead to subsampling algorithms that are quite efficient in some sense \cite{campbell2018giga}, but many popular simple precomputations cannot lead to asymptotic improvement \cite{JohndrowJamesE2020NFLf}, and doing sufficiently powerful pre-computations might well be harder than doing the original MCMC algorithm. The current analysis could be applied to pre-calculations that act only through re-weighting samples (with perturbations small enough that Lemma \ref{LemmaChatPertDrive} remained true for some other value of $C$), but the analysis of more general pre-calculations would require new perturbation arguments.
    
    \item \textbf{Complex Models:} The results in this paper, along with all other lower-bounds on errors of subsampling MCMC that we are aware of, show only that computational cost scales at least \textit{linearly} in the size $n$ of the dataset. This sort of result is important for the analysis of subsampling MCMC for ``typical" Bayesian problems, where naive MCMC scales at most linearly in $n$. However, it says very little about ``atypical" problems where no linear-in-$n$ algorithm is available. See the survey \cite{rudolf2024perturbationsmarkovchains} for several situations in which subsampling is known to give improvements for inference with complex models.
\end{enumerate}

\bibliographystyle{alpha}
\bibliography{ESCBib}

\newpage 

\begin{appendix}

\section{Verifying Perturbation Assumptions}

We give some concrete calculations for verifying the perturbation assumptions in certain special cases. Qualitatively-similar results are known for many other distributions, but (to our knowledge) the required calculations are longer and more abstract. See Chapter 3 of \cite{asymptotia} for our favorite introduction to such bounds.

\subsection{Verifying Assumption \ref{AssumptionCGE_Good}} \label{sec:technical-lemmas}

Assumption \ref{AssumptionCGE_Good} is unfamiliar, but seems to hold for a wide variety of popular models. In many cases, including Gaussians, this can be easily verified using simple moment bounds.

We recall the following sharp moment condition found in Theorem 1 of \cite{TVLB22}. For any distributions $P,Q$ on $\mathbb{R}$ with distinct means $\mu_{P}, \mu_{Q}$ and variances $\sigma_{P}^{2},\sigma_{Q}^{2}$, 
\begin{equation*}
\| P - Q \|_{\mathrm{TV}} \geq \frac{(\mu_{P} - \mu_{Q})^{2}}{ (\sigma_{P} + \sigma_{Q})^{2} + (\mu_{P} - \mu_{Q})^{2}}.
\end{equation*}

In the special case that $p(x|\theta)$ is a Gaussian with mean $\theta$ and fixed variance 1 and the prior $p$ is uninformative, this bound gives:

\begin{equation*}
\| p(\theta|(x_{1},\ldots,x_{n})) - p_{\delta_{n}}(\theta|(x_{1},\ldots,x_{n})) \|_{\mathrm{TV}} \geq \frac{\delta_{n}^{2}}{4n^{-1} + \delta_{n}^{2}}.
\end{equation*}

Since we are interested in the regime that $\delta_{n} \gg \frac{1}{\sqrt{n}}$, this lower bound converges to 1 as $n$ goes to infinity.

\subsection{Verifying Assumption \ref{AssumptionCGE_Good_Big_Time}}\label{secAppLongTime}

We first check that Assumption \ref{AssumptionCGE_Good_Big_Time} holds for Gaussians. Denote by $\mu_{1}, \mu_{2}$ two Gaussians with means $\theta_{1} < \theta_{2}$ and identical variance $1$. It is straightforward to check that $\mu_{1}$ has higher density than $\mu_{2}$ exactly on the set $A = \{ x \, : \, x < \frac{\theta_{1} + \theta_{2}}{2} \}$. Thus, denoting by $\Phi$ the CDF of a standard Gaussian, we have:

\begin{align*}
    \| \mu_{1} - \mu_{2} \|\mathrm{TV} &= 2(\mu_{1}(A) - \mu_{2}(A)) \\
    &= 2 (\Phi(\frac{\theta_{2} - \theta_{1}}{2}) - \Phi(-\frac{\theta_{2} - \theta_{1}}{2}) ) \\
    &= \frac{2}{\sqrt{2 \pi}} (\theta_{2}-\theta_{1}) + O((\theta_{2}-\theta_{1})^{2}).
\end{align*}

Thus, with an uninformative prior, Assumption \ref{AssumptionCGE_Good_Big_Time} is satisfied for any $0 < \gamma < \frac{2}{\sqrt{2 \pi}}$.

\end{appendix}

\end{document}